\documentclass[letterpaper,aps,amssymb,amsmath,amsfonts,twocolumn,pra]{revtex4}
\usepackage[latin1]{inputenc}
\usepackage{amsmath}
\usepackage{graphicx} 
\usepackage{subfigure} 

\usepackage{amsfonts}
\usepackage{hyperref} 
\usepackage{xcolor}
\usepackage{verbatim} 
\usepackage{amsthm}
\usepackage{multirow}

\newtheorem{prop}{PROPOSITION}[section]

\newtheorem{corol}{COROLLARY}[section]

\newtheorem{constr}{CONSTRUCTION}[section]

\begin{document}
\title{Mutually unbiased bases with free parameters}
\author{Dardo Goyeneche}
\email{dgoyeneche@cefop.udec.cl}
\affiliation{National Quantum Information Center of Gda\'{n}sk,  81-824 Sopot,  Poland}
\affiliation{Faculty of Applied Physics and Mathematics, Technical University of Gda\'{n}sk, 80-233 Gda\'{n}sk, Poland}
\affiliation{Departamento de Fis\'{i}ca, Universidad de Concepci\'{o}n, Casilla 160-C, Concepci\'{o}n, Chile\\Center for Optics and Photonics, Universidad de Concepci\'{o}n, Casilla 4012, Concepci\'{o}n, Chile}
\author{Santiago Gomez}
\affiliation{Departamento de Fis\'{i}ca, Universidad de Concepci\'{o}n, Casilla 160-C, Concepci\'{o}n, Chile\\Center for Optics and Photonics, Universidad de Concepci\'{o}n, Casilla 4012, Concepci\'{o}n, Chile}

\date{\today}

\begin{abstract}
We present a systematic method to introduce free parameters in sets of mutually unbiased bases. In particular, we demonstrate that any set of $m$ real mutually unbiased bases in dimension $N>2$ admits the introduction of $(m-1)N/2$ free parameters which cannot be absorbed by a global unitary operation. As consequence, there are $m=k+1$ mutually unbiased bases in every dimension $N=k^2$ with $k^3/2$ free parameters, where $k$ is even. We construct the maximal set of triplets of mutually unbiased bases for two-qubits systems and triplets, quadruplets and quintuplets of mutually unbiased bases with free parameters for three-qubits systems. Furthermore, we study the richness of the entanglement structure of such bases and we provide the quantum circuits required to implement all these bases with free parameters in the laboratory. Finally, we find the upper bound for the maximal number of real and complex mutually unbiased bases existing in every dimension. This proof is simple, short and it considers basic matrix algebra.
\end{abstract}
\maketitle

\section{Introduction}
Mutually unbiased bases (MUB) have an ubiquitous role in quantum mechanics. They are useful to generate quantum key distribution protocols \cite{Bennett,Brub,Cerf}, detection of entanglement \cite{Spengler}, quantum random access codes \cite{CGS08}, dense coding, teleportation, entanglement swapping and covariant cloning (see \cite{Durt} and references therein). Furthermore, a maximal set of MUB allows us to univocally reconstruct quantum states \cite{I81}. On the other hand, entropic certainty \cite{S95,RPZ14} and uncertainty relations \cite{WW10}, have been considered for MUB. Such important applications have motivated an enormous effort to understand the underlying structure behind incomplete \cite{MBGW14,G13} and complete \cite{I81,WF89} sets of MUB. In particular, incomplete sets of MUB have an important role in Bell inequalities \cite{BCPSW14}, uncertainty relations \cite{A04} and locking of classical correlations in quantum states \cite{DHLST04,BW07,DG09}. Despite of the important advance done for complete sets of MUB in prime \cite{I81} and prime power \cite{WF89} dimensions, incomplete sets of MUB seem to be much more challenging. Indeed, the full classification of all possible sets of MUB for 2-qubit systems has been recently done \cite{BWB10} and the multipartite case is poorly understood. The lack of a deeper understanding of mutually unbiased bases seems to be the absence of a suitable mathematical tool. Indeed, for three or more qubits systems it is not known the existence of  quadruplet of MUB having free parameters and a few triplets were accidentally found \cite{B94}. In this work, we enlighten this area of research by presenting a systematic method to introduce free parameters in incomplete sets of MUB. As consequence, we demonstrate that any set of $m$ real MUB existing in any dimension $N$ admits the introduction of free parameters with our method. We also show that our construction is not restricted to the consideration of real bases. Indeed, we illustrate our method by explicitly constructing the maximal set of triplets for 2-qubits systems and several triplets, quadruplets and quintuplets of MUB having free parameters for 3-qubits systems. All of these cases involve complex MUB. Furthermore, we analyze the entanglement structure of such sets of MUB and provide the quantum circuit required to implement all these sets in the laboratory.

This work is organized as follows: In Section II we present a short introduction to mutually unbiased bases, the link to complex Hadamard matrices and we resume the state of the art of MUB with free parameters. In Section III we present our method to introduce free parameters in incomplete sets of MUB. In Section IV we prove that any set of real MUB admits the introduction of the maximal number of parameters allowed by our method. In Section V we construct triplets, quadruplets and quintuplets of MUB having free parameters for three qubit systems and study the entanglement structure of each case. In Section VI we resume our main results, conclude and discuss some open questions. Additionally, we illustrate our method by explicitly solving the simplest case of triplets of MUB for 2-qubits systems (see Apendix \ref{appendix1}). The explicit construction of a quadruplet and a quintuplet of MUB having free parameters for three qubit systems is provided in Appendix \ref{appendix2}. In Appendix \ref{appendix3} we derive the quantum circuit required to generate every quadruplet and quintuplet of MUB presented in this work. Finally, in Appendix \ref{appendix4} we find a simple and short proof for the upper bound of the maximal number of real and complex MUB in every dimension by considering basic matrix algebra.

\section{Mutually unbiased bases and complex Hadamard matrices}

In this section, we present some fundamental properties about \emph{mutually unbiased bases} (MUB) and \emph{complex Hadamard matrices} (CHM) required to understand the rest of the work. For a complete review about MUB and CHM we suggest \cite{DEBZ10} and \cite{TZ06}, respectively. Two orthonormal bases $\{\phi_j\}_{j=0,\dots,N-1}$ and $\{\psi_k\}_{k=0,\dots,N-1}$ defined in $\mathbb{C}^N$ are \emph{mutually unbiased} if
\begin{equation}\label{MUB}
|\langle\phi_j,\psi_k\rangle|^2=\frac{1}{N},
\end{equation}
for every $j,k=0,\dots,N-1$. In general, a set of $m>2$ orthonormal bases are MUB if every pair of bases of the set is MUB. A set of $m$ MUB is called \emph{extensible} if there exist an $m+1$th basis which is mutually unbiased with respect to the rest of the bases. It has been shown that $m=N+1$ MUB exist for $N$ prime \cite{I81} and prime power \cite{WF89}. Additionally, maximal sets of MUB can be constructed in prime power dimensions by considering gaussian sums and finite fields \cite{RBKS05,KRBS09,KSSL09}. For any other dimension $N=p_1^{r_1}p_2^{r_2}\dots p_k^{r_k}$ ($p_1^{r_1}<p_2^{r_2}<\dots< p_k^{r_k}$) the maximal value for $m$ is not known and the lower bound $m\geq p_1^{r_1}$ is provided by the maximal number of fully separable (i.e., tensor product) MUB \cite{KR04}. Additionally, in dimensions of the form $N=k^2$ it is possible to find $m=k+1$ real MUB by considering orthogonal Latin squares \cite{PDB09}. Let us arrange the bases $\{\phi_j\}$ and $\{\psi_k\}$ in columns of unitary matrices $B_1$ and $B_2$, respectively. Thus, if $B_1$ and $B_2$ are MUB we have
\begin{equation}\label{B1B2H}
B_1^{\dag}B_2=H,
\end{equation}
where $H$ is a CHM. An $N\times N$ matrix $H$ is called a \emph{complex Hadamard matrix} (CHM) if it is unitary and all its complex entries have the same amplitude $1/\sqrt{N}$. For example, the Fourier matrix $(F_N)_{jk}=\frac{1}{\sqrt{N}}e^{2\pi ijk/N}$ is a CHM for every $N$, where $i=\sqrt{-1}$. Two CHM $H_1$ and $H_2$ are equivalent if there exists permutation matrices $P_1,P_2$ and diagonal unitary matrices $D_1,D_2$ such that $H_2=P_1D_1H_1D_2P_2$. Therefore, MUB and CHM are close related: any set of $m$ MUB $\mathcal{S}_1=\{B_1,\dots,B_m\}$ is unitary equivalent to a set $\mathcal{S}_2=\{\mathbb{I},H_1,\dots,H_{m-1}\}$, where $\mathbb{I}$ represents the computational basis and $H_1,\dots,H_{m-1}$ are CHM. Indeed, the unitary transformation that connects $\mathcal{S}_1$ with $\mathcal{S}_2$ is $B_1^{\dag}$. That is, 
\begin{eqnarray}\label{MUBCHM}
B_1^{\dag}(\mathcal{S}_1)&=&\{B_1^{\dag}B_1,B_1^{\dag}B_2,\dots,B_1^{\dag}B_m\}\nonumber\\
&=&\{\mathbb{I},H_1,\dots,H_{m-1}\}\nonumber\\
&=&\mathcal{S}_2,
\end{eqnarray}
where we used Eq.(\ref{B1B2H}). Alternatively, $B_k^{\dag}(\mathcal{S}_1)$ also provides an analogous result for $k=2,\dots,m$. The full classification of CHM and MUB has been solved up to dimension $N=5$ (see \cite{H96} and \cite{BWB10}, respectively). For $N=6$ both problems remain open despite a considerable effort made during the last 20 years \cite{Z99,A05,JMM10,G04,BH07,BBELTZ07,BW08,BW10,JMMSW09,G13}. The problems also remain open for any dimension $N>6$. For example, they are open in the prime dimension $N=7$, where a single 1-parametric family of complex Hadamard matrices is known \cite{P97} and a maximal set of 8 MUB is known \cite{I81} but incomplete sets of MUB are not yet characterized. Indeed, it is still open the question whether a triplet of MUB having free parameters exist in dimension $N=7$.

Let us summarize the state of the art about sets of MUB having free parameters. First, any set of $m$ MUB in prime dimension $N=p$ of the form $\{\mathbb{I}, F_p, C_1,\dots,C_{m-2}\}$ is \emph{isolated}, where $F_p$ is the Fourier matrix and $\{C_1,\dots,C_{m-2}\}$ are circulant CHM \cite{H08}.  A complex Hadamard matrix is isolated if there is no family of CHM connected with it \cite{TZ06}. By family we understand a set of inequivalent complex Hadamard matrices depending on some free parameters. We extend the same definition to sets of MUB: a set of $m$ MUB is isolated if there is no family of $m$ MUB connected with it. For example, any set of $m\leq N+1$ MUB in dimension $N=2,3$ and $N=5$ is isolated. In dimension $N=4$ there is a unique 3-parametric triplet of MUB of the form $\{\mathbb{I},F^{(1)}_4(x),H(y,z)\}$ and quadrupets and quintuplets of MUB are isolated \cite{BWB10}. In dimension $N=6$ a 1-parametric triplet of MUB exists \cite{Z99}. Moreover, two-parametric triplets of the form $\{\mathbb{I},F^{(2)}_6(x,y),H(x,y)\}$ exist for any $x,y\in[0,2\pi)$ and seem to be unextendible for any pair $x,y$ \cite{JMM10,G13}. Also, in dimensions $N=9$ \cite{Faugere} and $N=4k$ 1-parametric triplets of MUB can be defined by considering cyclic $n$-roots \cite{B94}.

All the above sets of MUB with free parameters were found by taking advantage of special properties holding in specific dimensions. The existence of quadruplets of MUB having free parameters \emph{is still unknown} in every dimension, as far as we know. In the next section we present a systematic method to introduce free parameters in sets of $m$ MUB in dimension $N$.

\section{MUB with free parameters}\label{SIII}
A set of $r>N$ vectors $\{v_k\}\subset\mathbb{C}^N$ has associated a Gram matrix $G\in\mathbb{C}^{r\times r}$, where $G_{ij}=\langle v_i,v_j\rangle$, $i,j=0,\dots,r-1$ and $\mathrm{Rank(G)}=N$. Reciprocally, from the Gram matrix $G$ of size $r$ and rank $N$ we can always find a set of vectors $\{v^{\prime}_k\}$ such that $G_{ij}=\langle v^{\prime}_i,v^{\prime}_j\rangle$ and $v^{\prime}_k\in\mathbb{C}^N$. The set of vectors $\{v_k\}$ and $\{v^{\prime}_k\}$, associated to the same $G$, are connected by means of a unitary transformation. Thus, they define the same geometrical structure in the complex projective space $\mathbf{CP}^{N-1}$. The vectors $\{v^{\prime}_k\}$ can be found from $G$ by considering the \emph{Cholesky decomposition}, i.e., to find the unique upper triangular matrix $L$ having positive diagonal entries such that $G=L^{\dag}L$. Thus, the $r$ vectors $\{v^{\prime}_k\}$ are given by the $r$ columns of $L$, where we only have to consider entries of the first $N$ rows of $L$ (the rest are null because of the rank restriction). In this work, we are particularly interested to study Gram matrices associated to a set of $m$ MUB $\{\mathbb{I},H_1,H_2,\dots,H_{m-1}\}$ in $\mathbb{C}^N$. That is,
\begin{equation}\label{Gram}
G=\left(
  \begin{array}{ccccc}
    \mathbb{I}     &       H_1       &       H_2       & \dots&  H_{m-1}  \\
    H_1^{\dag}     &  \mathbb{I}     &H_1^{\dag}H_2    & \dots& H_1^{\dag}H_{m-1} \\
    H_2^{\dag}     &   H_2^{\dag}H_1 &    \mathbb{I}   & \dots& H_2^{\dag}H_{m-1}  \\
    \vdots         &    \vdots       &       \vdots    &\ddots&    \vdots \\
    H_{m-1}^{\dag} &H_{m-1}^{\dag}H_1&H_{m-1}^{\dag}H_2&\dots &   \mathbb{I} \\
  \end{array}
\right).
\end{equation}
Note that this matrix naturally defines a structure of $m^2$ square blocks of size $N$, each of them defined by a unitary matrix of the form $H^{\dag}_iH_j$, where $i,j=0,\dots,m-1$ and $H_0=\mathbb{I}$. The Cholesky decomposition of this Gram matrix $G$ is given by
\begin{equation}\label{Choleski}
L=\left(
  \begin{array}{ccccc}
    \mathbb{I}     &       H_1       &       H_2       & \dots&  H_{m-1}  \\
    0_N     &  0_N     &0_N    & \dots& 0_N \\
    \vdots         &    \vdots       &       \vdots    &\ddots&    \vdots \\
   0_N &0_N&0_N&\dots & 0_N \\
  \end{array}
\right),
\end{equation}
where $0_N$ are zero matrices of size $N$. So, the set of $m$ MUB is clearly given by  $\{\mathbb{I},H_1,H_2,\dots,H_{m-1}\}$ which corresponds to the first block of rows of $G$ (see Eq.(\ref{Gram})). This important property \emph{substantially simplifies} our method. In a previous work, we found the most general way to introduce free parameters in pairs of columns (or rows) of any complex Hadamard matrix in every dimension \cite{G132}. A free parameter can be introduced in two columns $C_1$ and $C_2$ of a CHM if and only if $C_1\circ C_2\in\mathbb{R^N}$, which only holds for $N$ even.  Here, the circle denotes the (entrywise) Hadamard product, that is, $(C_1\circ C_2)_j=(C_1)_j(C_2)_j$, $j=0,\dots,N-1$. Pairs of columns (or rows) satisfying this property were called \emph{equivalent to real pairs} (ER pairs). The construction of the CHM having free parameters is very simple:
\begin{constr}[\cite{G132}]\label{construction}
Given an ER pair of columns $\{C_1,C_2\}$ we introduce a free phase $e^{i\alpha}$ in the $j$th entries $(C_1)_j$ and $(C_2)_j$ if $(C_1^*\circ C_2)_j<0$ for $j=0,\dots,N-1$.
\end{constr}
Here, the asterisk denotes complex conjugation. Note that $\sum_{j=0}^{N-1}(C_1^*\circ C_2)_j$ is the inner product between the column vectors $C_1$ and $C_2$, which has to be zero by definition of CHM. Let us exemplify this method by introducing two free parameter in the Fourier matrix
\begin{equation}
F_4=\frac{1}{2}\left(\begin{array}{cccc}
1&1&1&1\\
1&i&-1&-i\\
1&-1&1&-1\\
1&-i&-1&i
\end{array}
\right).
\end{equation}
That is,
\begin{equation}\label{FamF4}
F_4(\alpha,\beta)=\frac{1}{2}\left(\begin{array}{cccc}
1&1&1&1\\
1e^{i\alpha}&ie^{i\beta}&-1e^{i\alpha}&-ie^{i\beta}\\
1&-1&1&-1\\
1e^{i\alpha}&-ie^{i\beta}&-1e^{i\alpha}&ie^{i\beta}
\end{array}
\right).
\end{equation}
Here, we considered the ER pairs of columns $\{C_1,C_3\}$ and $\{C_2,C_4\}$ to introduce the paramaters $\alpha$ and $\beta$, respectively. Note that $\alpha$ and $\beta$ are \emph{aligned} in the same rows. A set of $N/2$ ER pairs producing aligned free parameters in matrices of size $N$ are called \emph{aligned ER pairs}. The remarkable property of aligned ER pairs is that one of the parameters is always linearly dependent (if we consider the equivalence of CHM defined above).  Thus, the parameter $\beta$ is linearly dependent and the Fourier family has only one relevant parameter \cite{TZ06}.

Let us now extend this method to the construction of MUB with free parameters. Here, the key ingredient is the generalization of the concept of ER pairs: a set of two columns $\{C_1,C_2\}$ of a Gram matrix of $m$ MUB in dimension $N$ is called a \emph{generalized ER} pair (GER) if $C_1\circ C_2\in\mathbb{R}^{mN}$. The following result, natural generalization of Construction \ref{construction}, is the main result of this work:
\begin{prop}\label{mainprop}
Let $G$ be the Gram matrix of a set of $m$ MUB in dimension $N$ and suppose that it has $\mathcal{N}$ GER pair of columns, where both vectors of each GER pair belong to the same block of columns. Then, the set of MUB admits the introduction of $\mathcal{N}$ free parameters.
\end{prop}
\begin{proof}
For simplicity let us first consider the case of 3 MUB ($m=3$) in dimension $N$, where the Gram matrix is given by
\begin{equation}\label{Gram2}
G=\left(
  \begin{array}{ccc}
    \mathbb{I}     &       \mathbf{H_1}       &       H_2\\
    H_1^{\dag}     &  \mathbf{I}     &H_1^{\dag}H_2\\
    H_2^{\dag}     &   \mathbf{H_2^{\dag}H_1} &    \mathbb{I}
  \end{array}
\right),
\end{equation}
and suppose that $G$ has a GER pair of columns $\{C_i,C_j\}$ in the same block of columns (i.e., $\mathrm{Int}[i/N]=\mathrm{Int}[j/N]$, where Int means integer part). Therefore, $N$ of the products $(C_i)_k(C_j)_k$ are zero because of the corresponding identity block $\mathbb{I}$ and only $2N$ values of these products play a role in $C_i\circ C_j$. Thus, a free parameter can be introduced in both columns $C_i$ and $C_j$ by applying Construction \ref{construction} to the $2N$ dimensional subvectors of $C_i$ and $C_j$ having $2N$ non-zero entries. Note that after introducing the parameter the hermiticity of $G$ is destroyed. In order to restore it we have to apply the same method to the GER pair of rows $\{R_i,R_j\}$ which always exists because of the hermiticity of $G$. This lead us to a 1-parametric set of matrices satisfying  \emph{i)} $G(\alpha)=U(\alpha)GU^{\dag}(\alpha)$ and \emph{ii)} $|G(\alpha)_{ij}|=|G_{ij}|$ for any $\alpha\in[0,2\pi)$. Note that \emph{i)} holds for any $U$ whereas \emph{ii)} is strongly dependent on our construction. Here, $U(\alpha)$ and $U^{\dag}(\alpha)$ represent the introduction of a free parameter in columns two and rows, respectively. Furthermore, $G(\alpha)$ and $G(0)$ have the same eigenvalues for any $\alpha\in[0,2\pi)$ so $G(\alpha)$ is a 1-parametric set of Gram matrices defining a 1-parametric set of $m$ MUB in dimension $N$. If $G$ has $\mathcal{N}$ GER pairs then we can introduce $\mathcal{N}$ free parameters in the same way. The generalization to any $m>3$ is straightforward from the above explanation.
\end{proof}
Let us emphasize the importance of considering both vectors of a GER pair in the same block of columns: suppose that we choose a GER pair formed by columns of different blocks (e.g. the 5th and 9th columns of the Gram matrix given in the example Eq.(\ref{G12})) and we introduce a free parameter. Despite of this action generates a genuine Gram matrix the set of 3 bases would be \emph{not} composed by MUB with free parameters. This is simple to understand because of the free parameter would appear in a single vector of the second and third basis.

In the particular case of $m=2$ our Proposition \ref{mainprop} is reduced to Construction \ref{construction}, which has been derived in a previous work \cite{G132}. That is, to introduce free parameters in a pair of MUB is equivalent to introduce free parameters in a CHM, as suggested by Eq.(\ref{MUBCHM}). It worths to mention that GER pairs of columns belonging to the first block of $G$ (i.e., $\{C_i,C_j\}$ with $i,j<N$) \emph{always} produce parameters that can be absorbed in global unitary transformations. Roughly speaking, in this case the parameters do not appear in the inner product of vectors of two different bases (see Eq.(\ref{Gram})). The explicit construction of the maximal set of triplets of MUB for 2-qubits systems is given in Appendix \ref{appendix1}. Also, a triplet, a quadruplet and a quintuplet of MUB having free parameters for 3-qubits systems are given in Appendix \ref{appendix2}. We encourage to the reader to have a close look to the examples in order to clearly understand our method.
 
\section{Families stemming from real MUB}
As we have shown, Proposition \ref{mainprop} allows us to introduce free parameters in Gram matrices of MUB having GER pairs. In this section, we demonstrate that every set of $m$ of real MUB in dimension $N>2$ allows the introduction of the maximal  number of parameters allowed by GER pairs:
\begin{prop}
Any set of $m$ real MUB in dimension $N>2$ admit the introduction of  $Nm/2$ free parameters. Furthermore, $(m-1)N/2$ of these parameters cannot be absorbed by global unitary transformations and, at most, one of them is linearly dependent.
\end{prop}
\begin{proof}
Every pair of columns belonging to the same block is clearly a GER pair. Therefore, there are $Nm/2$ GER pairs allowing the introduction of $Nm/2$ free parameters. The rest of the proof is straghtforward (already explained in the proof of Prop. \ref{mainprop}).
\end{proof}
Furthermore, note that there are many different ways to define the GER pairs and so many different families of MUB can be constructed. Precisely, there are $\binom{N}{2}$ different ways to define GER pairs in each of the $m-1$ blocks of columns (the first block only provides unitary equivalent MUB). That is, a total of $(m-1)N(N-1)/2$ different ways. We do not know how many ways are inequivalent for $N>4$. As we mentioned before, in dimensions $N=k^2$ it is possible to construct $m=k+1$ real MUB. By combining this result with the above proposition we have the following result:
\begin{corol}
In every dimension $N=k^2$ there exists $m=k+1$ MUB admitting $k^3/2$ free parameters.
\end{corol}
Here, we consider $k>1$ and thus $kN$ parameters cannot be absorbed by a global unitary transformation. In dimension $N=4$ there exists $m=3$ real MUB and thus we can introduce $k^3/2=4$ free independent parameters, where the GER pairs are aligned and thus one of the four parameters is linearly dependent (see Appendix \ref{appendix1}). Here, the 12 possible ways to introduce free parameters produce equivalent sets \cite{BWB10}. Such triplets are also equivalent to the solution found in Appendix \ref{appendix1}. Sets of real MUB are not the only cases where a maximal number of parameters can be introduced with our method. Indeed, in the next section we construct sets of MUB with free parameters in the 3-qubit space by considering complex MUB.

\section{MUB for the 3-qubit space}\label{MUB3Q}
In dimension $N=2^3$ there is a maximal number of 9 MUB. Indeed, four maximal sets having a different entanglement structure have been constructed \cite{RBKS05}:
\begin{eqnarray}\label{MUBAndrei}
\mathcal{S}_1&=&(2,3,4)\hspace{0.5cm}\mathcal{S}_2=(1,6,2)\nonumber\\
\mathcal{S}_3&=&(0,9,0)\hspace{0.5cm}\mathcal{S}_4=(3,0,6)
\end{eqnarray}
where the first, second and third entries denote the number of fully separable, biseparable and maximally entangled bases, respectively. These sets of MUB are not equivalent under Clifford operations but they are equivalent under general unitary transformations. In this section we focus on the construction of pairs, triplets, quadruplets and quintuplets of MUB having free parameters and stemming from elements of $\mathcal{S}_4$. We have chosen this particular set of MUB because it contains the highest number of maximally entangled bases and, consequently, it has potentially important applications in quantum information theory.
\begin{center}
\label{Tabla1}
\begin{table}
\begin{tabular}{c|c|c}

   \#MUB(m) & \# Param. & Example \\
  \hline
  2 & 3 & $\{\mathbb{I},H_1\}^*$\\
  2 & 4 & $\{\mathbb{I},H_1\}$\\
  3 & 7 & $\{\mathbb{I},H_1,H_2\}^*$ \\
  3 & 8 & $\{\mathbb{I},H_1,H_2\}$ \\
  4 & 4 & $\{\mathbb{I},H_1,H_2,H_3\}$ \\
  5 & 0 & $\{\mathbb{I},H_1,H_2,H_3,H_4\}$ \\
  5 & 4 & $\{\mathbb{I},H_1,H_2,H_3,H_5\}^*$ \\
  6-9 & 0 & $\{\mathbb{I},H_i,H_j,H_k,H_l,H_m\}$
\end{tabular}
\caption{MUB with free parameters in dimension $N=8$. The asterisk means that GER pairs are aligned, which produces one linearly dependent parameter. Note that a set of MUB in general allows many different choices for GER pairs, and some of them produce non-alligned GER pairs.  In Appendix \ref{appendix2} we detailedly explain all possible sets of $m$ MUB having free parameters.}
\end{table}
\end{center}
From considering Prop. \ref{mainprop} we find the following results for every subset of $m\leq 9$ MUB of $\mathcal{S}_4$:
\begin{itemize}
\item [\emph{(i)}] Every pair of MUB $\{\mathbb{I},H_i\}\subset\mathcal{S}_4$ admits the introduction of 4 free parameters, where $i=1,\dots,8$. 
\item [\emph{(ii)}] Every triplet of MUB $\{\mathbb{I},H_i,H_j\}\subset\mathcal{S}_4$ admits the introduction of 8 free parameters for every $i\neq j=1,\dots8$.
\item [\emph{(iii)}] Some quadruplets of MUB $\{\mathbb{I},H_j,H_k,H_l\}\subset\mathcal{S}_4$ admit the introduction of 4 free parameters in \emph{only one} of the bases, whereas the rest of the quadruplets do not admit free parameters.
\item [\emph{(iv)}] Every quadruplet admitting 4 free parameters can be extended to a quintuplet of MUB having 4 free parameters. The extension of quadruplets to quintuplets is not unique.
\item [\emph{(v)}] Every set of $6\leq m\leq9$ MUB do not admit free parameters.
\end{itemize}
The maximal number of parameters that can be introduced for every $m$ is provided in Table \ref{Tabla2} of Appendix \ref{appendix2}, where we present the proof of the above results.

The free parameters of all quadruplets and quintuplets of MUB described in \emph{(iii)} and \emph{(iv)} can be generated in the laboratory by considering 7 different quantum circuits (56 cases, 8 cases per circuit; see Appendix \ref{appendix3}), which involve local and Toffoli gates. The generation of the fixed set of bases, i.e. the set $\mathcal{S}_4$, requires different a circuit \cite{SSL14} which involves local, non-local controlled-phase and Toffoli gates \cite{B95}. Therefore, the sets of MUB with free parameters are generated by a composition of two different quantum circuits. The explicit expression of the quantum circuits is provided in Appendix \ref{appendix3}. 

The entanglement structure of the sets of MUB presented in Appendix \ref{appendix3} is very interesting. For example, let us consider the quintuplet $\{\mathbb{I},H_1(\alpha),H_2,H_3,H_5\}\subset\mathcal{S}_4$ (explicitly constructed in Appendix \ref{appendix2}). For simplicity, let us assume that the 4 parameters $\alpha_1$ to $\alpha_4$ are identical ($\alpha$). Here, $H_1(0)$ is a maximally entangled basis, in the sense that every vector of the basis is equivalent (i.e., up to local unitary operations) to the GHZ state $|GHZ\rangle=(|000\rangle+|111\rangle)/\sqrt{2}$. On the other hand, $H_1(\pi/2)$ is a biseparable basis. Indeed, every vector of the basis is equivalent to $|\phi\rangle=|0\rangle(|00\rangle+|11\rangle)/\sqrt{2}$. That is, Alice is separated and Bob and Charlie share a maximally entangled Bell state. For any $0<\alpha<\pi/2$ we have an intermediate amount of entanglement between Alice and Bob-Charlie whereas Bob and Charlie are as entangled as possible for any $\alpha$. That is, two of the three parties (Bob and Charlie) saturate the maximal amount of entanglement allowed by the monogamy of entanglement of three-partite systems. Indeed, for any value of the parameter $\alpha$ the single qubit reductions $\rho_{B}$ and $\rho_{C}$ are maximally mixed. \emph{It is highly non-trivial the fact that the basis $H_1(0)$ (maximally entangled) and $H_1(\pi/2)$ (fully separable) can be continuosly connected without loosing the unbiasity of the quintuplet of MUB for any value $\alpha\in[0,\pi/2)$}. In Appendix \ref{appendix3} we show that, in this case, the parameter $\alpha$ is fully controlled by local unitary operations generated by Bob (see quantum circuit G).  This means that Bob has full control on the entanglement existing between Alice and Bob-Charlie when we are restricted to \emph{keep the unbiasity} of the quintuplet.  Table \ref{Tabla3} in Appendix \ref{appendix3} shows all possible entanglement structures that can be find from quadruplets and quintuplets with free parameters arising from $\mathcal{S}_4$. The purity of the reductions $\rho_A$, $\rho_B$ and $\rho_C$ as a function of the free parameter $\alpha$ for the above quintuplet is depicted in Fig.(\ref{Fig1}).

\begin{figure}
\begin{center}
\scalebox{0.2}{\includegraphics{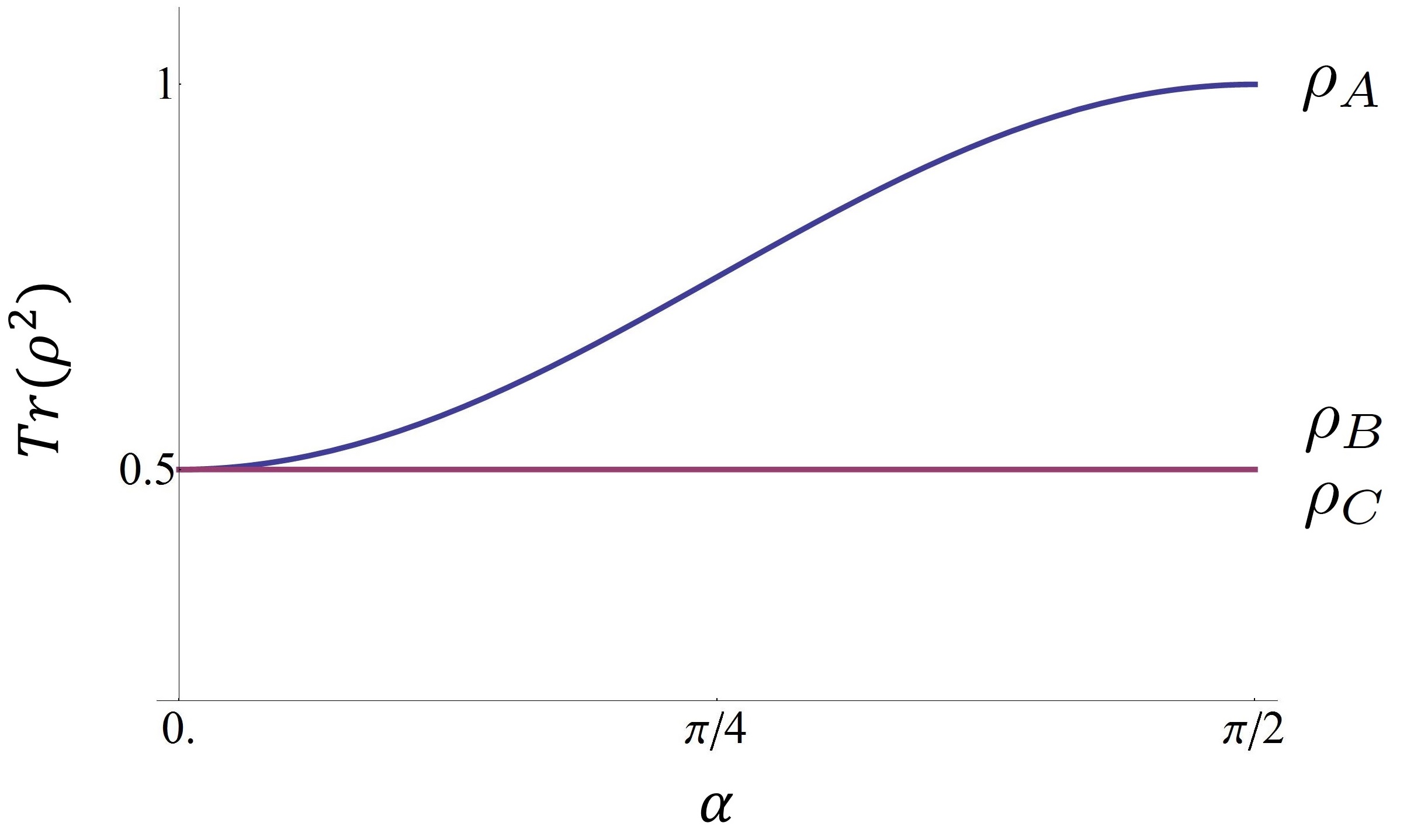}}
\end{center}
\caption{Purity of the reductions to Alice ($\rho_A$), Bob ($\rho_B$) and Charlie ($\rho_C$) for all the states of the basis $H_1(\alpha)$. As we can see, Bob and Charlie are as entangled as they can for every $\alpha$ whereas Alice is maximally entangled with Bob-Charlie for $\alpha=0$ (GHZ state) and separated for $\alpha=\pi/2$, where Bob and Charlie share a Bell state.}
\label{Fig1}
\end{figure}

\section{Summary, conclusion and discussion}
We have presented a systematic way to introduce free parameters in sets of $m$ mutually unbiased bases in dimension $N$ (see Proposition \ref{mainprop}). In particular, for $m=2$ our method is reduced to introduce free parameters in complex Hadamard matrices (see Construction \ref{construction} and also our previous work \cite{G132}). We proved that any set of $m$ real mutually unbiased bases existing in any dimension $N>2$ admit the introduction of free parameters. Furthermore, in every dimension $N=k^2$ there are $k+1$ mutually unbiased with $kN/2$ free parameters, where $k$ is even. We have found the maximal set of triplets of mutually unbiased bases with free parameters for two qubit-systems (see Appendix \ref{appendix1}). Also, we constructed pairs, triplets, quadruplets and quintuples of mutually unbiased bases having free parameters for 3-qubit systems (see Appendix \ref{appendix2}). Such sets are constructed from subsets of the maximal set of complex MUB $\mathcal{S}_4$  (see Eq.(\ref{MUBAndrei})). That is, our construction is not restricted to sets of real MUB. Additionally, we provided the complete set of quantum circuits required to implement all such quadruplets and quintuplets (see Appendix \ref{appendix3}). Finally, we presented a new proof for the upper bound of the maximal number of real and complex mutually unbiased bases existing in every dimension. This short and simple proof only involves basic algebra (see Appendix \ref{appendix4}).

The analysis provided in Section \ref{MUB3Q} for 3-qubits can be easily extended to a higher number of qubits. In order to do this we have to consider the construction of maximal sets of mutually unbiased arising from Galois fields \cite{ RBKS05}. Such construction is a generalization of the set $\mathcal{S}_4$, where every basis is a real Hadamard matrix multiplied by a diagonal unitary matrix containing $4th$ roots of the unity. In such cases we can always define GER pairs (see Prop. \ref{mainprop}) and, therefore, introduce free parameters in subsets of $m$ MUB. Finally, let us present some open issues: (\emph{i}) Find the subset of triplets, quadruplets and quintuplets of MUB considered in Appendix \ref{appendix3} such that they are extendible to 9 MUB, and (\emph{ii}) Is it possible to construct maximal sets of MUB with free parameters in some dimension? This question seems to have a negative answer for every $N$. However, a formal proof is only known for every $N\leq5$ \cite{BWB10}.

\section{Acknowledgements}
We thank to Luis Sanchez Soto and Markus Grassl for fruitfull discussions and Joel Tropp for his comments concerning rank inequalities and Hadamard product. We also thank to the Max Planck Institute for the Science of Light, where this work was partially done. DG is also thankful to Pawe{\l} Horodecki for the hospitality during his stay in Sopot, where this project was finished. This work was supported by FONDECyT Scholarship Nro 3120066 and PIA-CONICYT PFB0824 (Chile) and the ERC Advanced Grant  QOLAPS coordinated by Ryszard Horodecki (Poland).

\appendix

\section{MUB in dimension four}\label{appendix1}
Let us consider the simplest case where our method can be applied. The construction of $m=2$ MUB having free parameters is reduced to find a family of CHM and, thus, here we consider $m=3$. In dimensions $N=2$ and $N=3$ complex Hadamard matrices are isolated and, consequently, any set of MUB in such dimensions is isolated. On the other hand, in dimension 4 there is a family of CHM (see Eq.(\ref{FamF4})). So, the first case corresponds to $N=4$.  A fixed triplet of MUB for $N=4$ is given by:
\begin{equation}
H_1=\left(
  \begin{array}{cccc}
1&0&0&0\\
0&1&0&0\\
0&0&1&0\\
0&0&0&1
  \end{array}
\right),
\end{equation}
\begin{equation}
H_2=\frac{1}{\sqrt{2}}\left(
  \begin{array}{cccc}
1&1&1&1\\
1&1&-1&-1\\
1&-1&i&-i\\
1&-1&-i&i
  \end{array}
\right),
\end{equation}
\begin{equation}
H_3=\frac{1}{\sqrt{2}}\left(
  \begin{array}{cccc}
1&1&1&1\\
1&1&-1&-1\\
-1&1&1&-1\\
1&-1&1&-1
  \end{array}
\right).
\end{equation}
The Gram matrix $G$ associated to this set is given by $1/2$ of the following matrix
\begin{equation}\label{G12}
\left(
  \begin{array}{cccccccccccc}
2& 0& 0& 0& 1& 1& 1& 1& 1& 1& 1& 1\\
0& 2& 0& 0& 1& 1& -1& -1& 1& 1& -1& -1\\
0& 0& 2& 0& 1& -1& i& -i& -1& 1& 1& -1\\
0& 0& 0& 2& 1&-1& -i& i& 1& -1& 1& -1\\
1& 1& 1& 1& 2& 0& 0& 0& 1& 1& 1& -1\\
1& 1& -1& -1& 0& 2& 0& 0& 1& 1& -1& 1\\
1& -1& -i& i& 0& 0& 2& 0& i& -i& 1& 1\\
1& -1& i& -i& 0& 0& 0& 2& -i& i& 1& 1\\
1& 1& -1& 1& 1& 1& -i& i& 2& 0& 0& 0\\
1& 1& 1& -1& 1& 1& i& -i& 0& 2& 0& 0\\
1& -1& 1& 1& 1& -1& 1& 1& 0& 0& 2& 0\\
1& -1& -1& -1& -1& 1& 1& 1& 0& 0& 0& 2
  \end{array}
\right).
\end{equation}
Here, we have 6 GER pairs of columns and rows given by \{1-2;3-4;5-6;7-8;9-10;11-12\}. 
Note that the perfect match between GER pairs of columns and rows is given to the fact that $G$ is hermitian. The introduction of free parameters into the GER pairs of columns $\{C_1,C_2\};\{C_3,C_4\}$ and rows $\{R_1,R_2\};\{R_3,R_4\}$ implies that $H_2^{\dag}H_3$ and $H_3^{\dag}H_2$ do not depend on the parameters (See Eq.(\ref{Gram2})). As consequence, these MUB are unitary equivalent for any value of the two parameters. Thereby, they are not interesting for us. In general, we do not introduce free parameters in the first $N$ columns (and rows) of Gram matrices of $m$ MUB in dimension $N$ as we already explained in Section \ref{SIII}. From considering the remaining 4 GER pairs above defined we easily generate the following MUB with free parameters:
\begin{equation}\label{H2}
H_2(\alpha,\beta)=\frac{1}{\sqrt{2}}\left(
  \begin{array}{cccc}
1&1&1&1\\
1&1&-1&-1\\
e^{i\alpha}&-e^{i\alpha}&ie^{i\beta}&-ie^{i\beta}\\
e^{i\alpha}&-e^{i\alpha}&-ie^{i\beta}&ie^{i\beta}
  \end{array}
\right)
\end{equation}
\begin{equation}\label{H3}
H_3(\gamma,\delta)=\frac{1}{\sqrt{2}}\left(
  \begin{array}{cccc}
1&1&1&1\\
1&1&-1&-1\\
-e^{i\gamma}&e^{i\gamma}&e^{i\delta}&-e^{i\delta}\\
e^{i\gamma}&-e^{i\gamma}&e^{i\delta}&-e^{i\delta}
  \end{array}
\right).
\end{equation}
Note that our method can be considerably simplified by introducing the free parameters in a reduced region of the Gram matrix (\ref{G12}). This is because $G$ contains much more information than the set of 3 MUB (see Eq. \ref{Choleski}). Precisely, we can restrict our attention to introduce free parameters in the first 4 rows according to the existing pairs of GER and Proposition \ref{mainprop}. The rest of the rows give us the explicit expression of the inner products between the elements of the different bases. In general, we can restrict our attention to introduce parameters in the first $N$ rows of $G$ when we consider $m$ MUB in dimension $N$. Note that this implies to only consider the GER pairs of the last $(m-1)N$ columns (without considering GER pairs of rows).

The four parameters $\alpha,\beta,\gamma,\delta$ appearing in Eq.(\ref{H2}) and Eq.(\ref{H3}) cannot be absorbed by global unitary operations. However, they are aligned, so one of them can be absorbed in a global phase of a vector of the canonical basis $H_1$. Therefore, we find the following 3-parametric continuos triplet of MUB in dimension $N=4$:
\begin{equation}
\{\mathbb{I},H_2(\alpha,0),H_3(\gamma,\delta)\}.
\end{equation}
This has been reported as the most general triplet of MUB that can be constructed in dimension $N=4$. Quadruplets and quintuplets of MUB do not allow free parameters in dimension 4 and, consequently, they are isolated \cite{BWB10}.

\section{MUB in dimension 8}\label{appendix2}
In this appendix we construct the maximal number of triplets, quadruplets and quintuplets of MUB having free parameters from  $\mathcal{S}_4$ (see Eq.(\ref{MUBAndrei})). The key result is provided in Table \ref{Tabla2} where we present the complete set of GER pairs for $\mathcal{S}_4$.
\begin{center}\label{Tabla2}
\begin{table*}
\scalebox{1}[1]{
\begin{tabular}{| c | c | c | c | c | c | c | c | c |}
	\hline
	& $H_{1}$ & $H_{2}$ & $H_{3}$ & $H_{4}$ & $H_{5}$ & $H_{6}$ & $H_{7}$ & $H_{8}$\\
	\hline
	 $H_{1}^{\dagger}$& - & 1-3-4-8;2-5-6-7 & 1-4-5-7;2-3-6-8& 1-2-4-7;3-5-6-8  &1-4-5-6;2-3-7-8&1-6-7-8;2-3-4-5  &1-2-3-7;4-5-6-8 & 1-2-5-8;3-4-6-7 \\
    \hline 
	 $H_{2}^{\dagger}$& 1-2-5-8;3-4-6-7 & - & 1-3-4-8;2-5-6-7 & 1-2-3-6;4-5-7-8&1-2-5-8;3-4-6-7&1-2-5-8;3-4-6-7&1-4-5-7;2-3-6-8
	 &1-2-4-6;3-5-7-8\\
   \hline
	 $H_{3}^{\dagger}$& 1-2-4-7;3-5-6-8 & 1-6-7-8;2-3-4-5 &-&1-4-5-6;2-3-7-8&1-6-7-8;2-3-4-5  &1-3-4-8;2-5-6-7&1-2-4-8;3-5-6-7&1-3-4-8;2-5-6-7\\
    \hline 
   	$H_{4}^{\dagger}$& 1-3-5-7;2-4-6-8 &1-4-5-7;2-3-6-8 &1-2-4-6;3-5-7-8& -  & 1-2-3-6;4-5-7-8  &1-4-5-6;2-3-7-8&1-3-5-8;2-4-6-7&1-4-5-7;2-3-6-8\\
    \hline 
$H_{5}^{\dagger}$& 1-2-3-6;4-5-7-8 & 1-2-4-6;3-5-7-8 &1-2-5-8;3-4-6-7&1-3-5-7;2-4-6-8& -& 1-3-5-7;2-4-6-8&1-6-7-8;2-3-4-5&1-6-7-8;2-3-4-5\\
    \hline 
$H_{6}^{\dagger}$& 1-3-4-8;2-5-6-7 &1-2-3-7;4-5-6-8 & 1-3-5-6;2-4-7-8&1-2-5-8;3-4-6-7& 1-3-4-8;2-5-6-7 & -&1-2-5-6;3-4-7-8&1-3-5-6;2-4-7-8\\
    \hline   			      
$H_{7}^{\dagger}$& 1-6-7-8;2-3-4-5 & 1-2-5-8;3-4-6-7 & 1-6-7-8;2-3-4-5 &1-6-7-8;2-3-4-5&1-2-4-7;3-5-6-8& 1-2-3-6;4-5-7-8 &-& 1-2-3-7;4-5-6-8 \\
\hline
$H_{8}^{\dagger}$& 1-4-5-6;2-3-7-8 & 1-3-5-6;2-4-7-8 & 1-2-3-7;4-5-6-8 &1-3-4-8;2-5-6-7& 1-3-5-7;2-4-6-8 & 1-2-4-7;3-5-6-8 &1-3-4-6;2-5-7-8& - \\
\hline
 \end{tabular}}
\caption{GER pairs required to construct any subset of $m$ MUB steeming from the maximal set of 9 MUB $\mathcal{S}_4.$}
\end{table*}
\end{center}
How to read Table \ref{Tabla2}:
\begin{itemize}
\item[\emph{(i)}] The cell associated to column $H_k$ and row $H_j^{\dag}$ contains all the GER pairs allowed by the triplet $\{\mathbb{I},H_j,H_k\}$.
\item[\emph{(ii)}] The notation $i$-$j$-$k$-$l$ means that every possible combination of 2 non-repeated indices determine a GER pairs; that is, $\{C_i,C_j\}$, $\{C_i,C_k\}$, $\{C_i,C_l\}$, $\{C_j,C_k\}$, $\{C_j,C_l\}$ and $\{C_k,C_l\}$ are GER pairs.
\item[\emph{(iii)}] The semicolon (;) separates complementary sets of GER pairs (i.e., for $\{i$-$j$-$k$-$l$;$\mu$-$\nu$-$\kappa$-$\eta\}$ mixtures of Graco-Latin indices \emph{do not} form GER pairs).
\end{itemize}
To construct quadruplets or quintuplets of MUB we have to find the intersection of the sets of GER allowed by all the subsets of triplets. If there is no intersection then free parameters cannot be introduced. Let us construct a triplet of MUB:

 Suppose we want to introduce free parameters in the triplet $\{H_1,H_2,H_3\}$ (see Eq.(\ref{Gram2})). In order to introduce free parameters in $H_1$ we have to find common GER pairs in the cells associated to $H_2^{\dag}H_1$ (i.e., column 2, row 3 of Table \ref{Tabla2}: \{1-2-5-8;3-4-6-7\}) and $H_3^{\dag}H_1$ (i.e., column 2, row 4: \{1-2-4-7;3-5-6-8\}). This is equivalent to find GER pairs  appearing in the Gram matrix of $\{H_1,H_2,H_3\}$. Thus, the unique set of common GER pairs are given by $\{C_1,C_2\}$, $\{C_3,C_6\}$, $\{C_4,C_7\}$ and $\{C_5,C_8\}$. Analogously, we can find the GER pairs for the second and third block of $G$; that is, $\{C_1,C_8\}$, $\{C_2,C_5\}$, $\{C_3,C_4\}$, $\{C_6,C_7\}$ and $\{C_1,C_4\}$, $\{C_2,C_6\}$, $\{C_3,C_8\}$, $\{C_5,C_7\}$, respectively. Thus, we are in conditions to introduce 12 free parameters in the Gram matrix of the fixed set $\{H_1,H_2,H_3\}$. As noted in Appendix \ref{appendix1}, the Cholesky decomposition  allows us to simplify the introduction of free parameters by only considering the first $N$ rows of G. Thus, our 12 parametric set of $m=3$ MUB is given by
\begin{widetext}
\begin{equation}
H_1(\alpha_1,\alpha_2,\alpha_3,\alpha_4)=\left(\begin{array}{cccccccc}
-i&-i&i&i&-i&i&i&-i\\
-i&-i&-i&i&i&-i&i&i\\
-ie^{i\alpha_1}&ie^{i\alpha_1}&ie^{i\alpha_2}&-ie^{i\alpha_3}&-ie^{i\alpha_4}&-ie^{i\alpha_2}&ie^{i\alpha_3}&ie^{i\alpha_4}\\
ie^{i\alpha_1}&-ie^{i\alpha_1}&ie^{i\alpha_2}&ie^{i\alpha_3}&-ie^{i\alpha_4}&-ie^{i\alpha_2}&-ie^{i\alpha_3}&ie^{i\alpha_4}\\
e^{i\alpha_1}&-e^{i\alpha_1}&-e^{i\alpha_2}&-e^{i\alpha_3}&-e^{i\alpha_4}&e^{i\alpha_2}&e^{i\alpha_3}&e^{i\alpha_4}\\
e^{i\alpha_1}&-e^{i\alpha_1}&e^{i\alpha_2}&-e^{i\alpha_3}&1&-e^{i\alpha_4}e^{i\alpha_2}&e^{i\alpha_3}&-e^{i\alpha_4}\\
-1&-1&1&-1&1&1&-1&1\\
1&1&1&1&1&1&1&1
\end{array}
\right),
\end{equation}
\begin{equation}
H_2(\beta_1,\beta_2,\beta_3,\beta_4)=\left(\begin{array}{cccccccc}
-1&1&-1&-1&1&1&1&-1\\
-ie^{i\beta_1}&-ie^{i\beta_2}&-ie^{i\beta_3}&ie^{i\beta_3}&ie^{i\beta_2}&ie^{i\beta_4}&-ie^{i\beta_4}&ie^{i\beta_1}\\
-i&i&i&i&i&-i&-i&-i\\
-e^{i\beta_1}&-e^{i\beta_2}&e^{i\beta_3}&-e^{i\beta_3}&e^{i\beta_2}&-e^{i\beta_4}&e^{i\beta_4}&e^{i\beta_1}\\
-e^{i\beta_1}&e^{i\beta_2}&e^{i\beta_3}&-e^{i\beta_3}&-e^{i\beta_2}&e^{i\beta_4}&-e^{i\beta_4}&e^{i\beta_1}\\
-i&-i&i&i&-i&i&i&-i\\
ie^{i\beta_1}&-ie^{i\beta_2}&ie^{i\beta_3}&-ie^{i\beta_3}&ie^{i\beta_2}&ie^{i\beta_4}&-ie^{i\beta_4}&-ie^{i\beta_1}\\
1&1&1&1&1&1&1&1
\end{array}
\right),
\end{equation}
and
\begin{equation}
H_3(\gamma_1,\gamma_2,\gamma_3,\gamma_4)=\left(\begin{array}{cccccccc}
ie^{i\gamma_1}&ie^{i\gamma_2}&ie^{i\gamma_3}&-ie^{i\gamma_1}&-ie^{i\gamma_4}&-ie^{i\gamma_2}&ie^{i\gamma_4}&-ie^{i\gamma_3}\\
-e^{i\gamma_1}&e^{i\gamma_2}&e^{i\gamma_3}&e^{i\gamma_1}&e^{i\gamma_4}&-e^{i\gamma_2}&-e^{i\gamma_4}&-e^{i\gamma_3}\\
-1&1&-1&-1&1&1&1&-1\\
-i&-i&i&-i&i&-i&i&i\\
-e^{i\gamma_1}&-e^{i\gamma_2}&e^{i\gamma_3}&e^{i\gamma_1}&-e^{i\gamma_4}&e^{i\gamma_2}&e^{i\gamma_4}&-e^{i\gamma_3}\\
-ie^{i\gamma_1}&ie^{i\gamma_2}&-ie^{i\gamma_3}&ie^{i\gamma_1}&-ie^{i\gamma_4}&-ie^{i\gamma_2}&ie^{i\gamma_4}&ie^{i\gamma_3}\\
-i&i&i&-i&-i&i&-i&i\\
1&1&1&1&1&1&1&1
\end{array}
\right).
\end{equation}
\end{widetext}
Here, only 8 parameters are relevant because 4 of them do not appear in the inner products, and thus they can be absorbed by a global rotation (i.e., or $\alpha$'s or $\beta$'s or $\gamma$'s can be considered as zero without loosing of generality). Moreover, given that the parameters are aligned we have 7 independent parameters. This triplet can be straightforwardly extended to a quadruplet of MUB by adding the computational basis $H_9$. In order to construct a quintuplet we have to find a suitable extra basis. One way to do this is by considering $H_5$. For this choice we have a non-empty set of GER and 4 parameters can be introduced in $H_1$ (see Table \ref{Tabla2}, column starting with $H_1$). The remaining four bases $\mathbb{I},H_2,H_3$ and $H_5$ are fixed, where $H_2=H_2(0,0,0,0)$, $H_3=H_3(0,0,0,0)$ and $H_5$ is given by
\begin{equation}
H_5=\left(\begin{array}{cccccccc}
1&-1&1&1&-1&-1&-1&1\\
1&-1&-1&1&1&1&-1&-1\\
-i&-i&-i&i&i&-i&i&i\\
-i&-i&i&i&-i&i&i&-i\\
i&-i&-i&-i&-i&i&i&i\\
-i&i&-i&i&-i&i&-i&i\\
-1&-1&1&-1&1&1&-1&1\\
1&1&1&1&1&1&1&1\\\end{array}
\right).
\end{equation}
We encourage to the reader to verify that there is no intersection between the set o GER pairs for $H_2$, $H_3$ and $H_5$ (see Table \ref{Tabla2}) and, consequently, no more free parameters can be introduced. Therefore, the 4 parametric quintuplet is given by
\begin{equation}
\{\mathbb{I},H_1(\alpha_1,\alpha_2,\alpha_3,\alpha_4),H_2,H_3,H_5\}.
\end{equation}

\begin{center}
\begin{table}[h]
\scalebox{1}[1]{
\begin{tabular}{c|c|c|c|c|c|c}
Circ. A&Circ. B&Circ. C&Circ. D&Circ. E&Circ. F&Circ. G\\ \hline
12{\color{red}\textbf{4}}6&12{\color{red}\textbf{3}}4&123{\color{red}\textbf{6}}&134{\color{red}\textbf{5}}&1{\color{red}\textbf{2}}45&1{\color{red}\textbf{2}}37&{\color{red}\textbf{1}}235\\
125{\color{red}\textbf{7}}&127{\color{red}\textbf{8}}&12{\color{red}\textbf{5}}8&136{\color{red}\textbf{7}}&{\color{red}\textbf{1}}267&{\color{red}\textbf{1}}248&1{\color{red}\textbf{2}}68\\
{\color{red}\textbf{1}}347&{\color{red}\textbf{1}}368&13{\color{red}\textbf{4}}8&146{\color{red}\textbf{8}}&14{\color{red}\textbf{7}}8&135{\color{red}\textbf{8}}&1{\color{red}\textbf{3}}78\\
1{\color{red}\textbf{3}}56&14{\color{red}\textbf{6}}7&{\color{red}\textbf{1}}456&{\color{red}\textbf{1}}578&15{\color{red}\textbf{6}}8&1{\color{red}\textbf{4}}57&1{\color{red}\textbf{5}}67\\
{\color{red}\textbf{2}}478&{\color{red}\textbf{2}}358&23{\color{red}\textbf{7}}8&{\color{red}\textbf{2}}346&23{\color{red}\textbf{4}}7&2{\color{red}\textbf{3}}68&234{\color{red}\textbf{8}}\\
256{\color{red}\textbf{8}}&24{\color{red}\textbf{5}}7&{\color{red}\textbf{2}}567&2{\color{red}\textbf{3}}57&23{\color{red}\textbf{5}}6&246{\color{red}\textbf{7}}&245{\color{red}\textbf{6}}\\
34{\color{red}\textbf{6}}8&3{\color{red}\textbf{4}}56&{\color{red}\textbf{3}}467&2{\color{red}\textbf{4}}58&{\color{red}\textbf{3}}458&35{\color{red}\textbf{6}}7&345{\color{red}\textbf{7}}\\
3{\color{red}\textbf{5}}78&56{\color{red}\textbf{7}}8&457{\color{red}\textbf{8}}&2{\color{red}\textbf{6}}78&367{\color{red}\textbf{8}}&3{\color{red}\textbf{5}}68&{\color{red}\textbf{4}}678
 \end{tabular}}
\label{Tabla3}
\caption{Quantum circuit required to construct every quadruplet and quintuplet of MUB steeming from $\mathcal{S}_4$. The 4 numbers $ijkl$ denote the quintuplet $\{\mathbb{I},H_i,H_j,H_k,H_l\}$ whereas the red-bold number (online version) denotes which basis carries the parameters. Quadruplets are constructed by removing any basis from quintuplets.}
\end{table}
\end{center}

\section{Quantum circuits for quadruplets and quintuplets of MUB}\label{appendix3}
Every quadruplet and quintuplet of MUB with free parameters in dimension 8 shown in this work was generated from considering subsets of $\mathcal{S}_4$ (see Eq.(\ref{MUBAndrei})). Such sets  can be implemented in the laboratory by considering suitable quantum circuits. First, we have to consider the generation of the fixed bases $\mathcal{S}_4$ and then the introduction of the parameters. Therefore, the full quantum circuit is the composition of two circuits. The generation of the set $\mathcal{S}_4$ is given by the quantum circuit depicted in Figure \ref{CCLuis}. This circuit was recently derived \cite{SSL14}. The free parameters can be introduced by considering 7 quantum circuits (A to G). Table \ref{Tabla3} shows the quantum circuit required to generate every quadruplet and quintuplet of MUB. Here, every set of four numbers $ijkl$ denote the 4 CHM $H_i, H_j, H_k$ and $H_l$. The fifth basis is $H_9=\mathbb{I}$ which is implicit in the table. Thus, the quintuplet associated to $ijkl$ is given by $\{\mathbb{I},H_i,H_j,H_k,H_l\}$. From removing any basis we get a quadruplet having free parameters. The 7 quantum circuits are given by
\begin{figure}
\includegraphics[width=8cm]{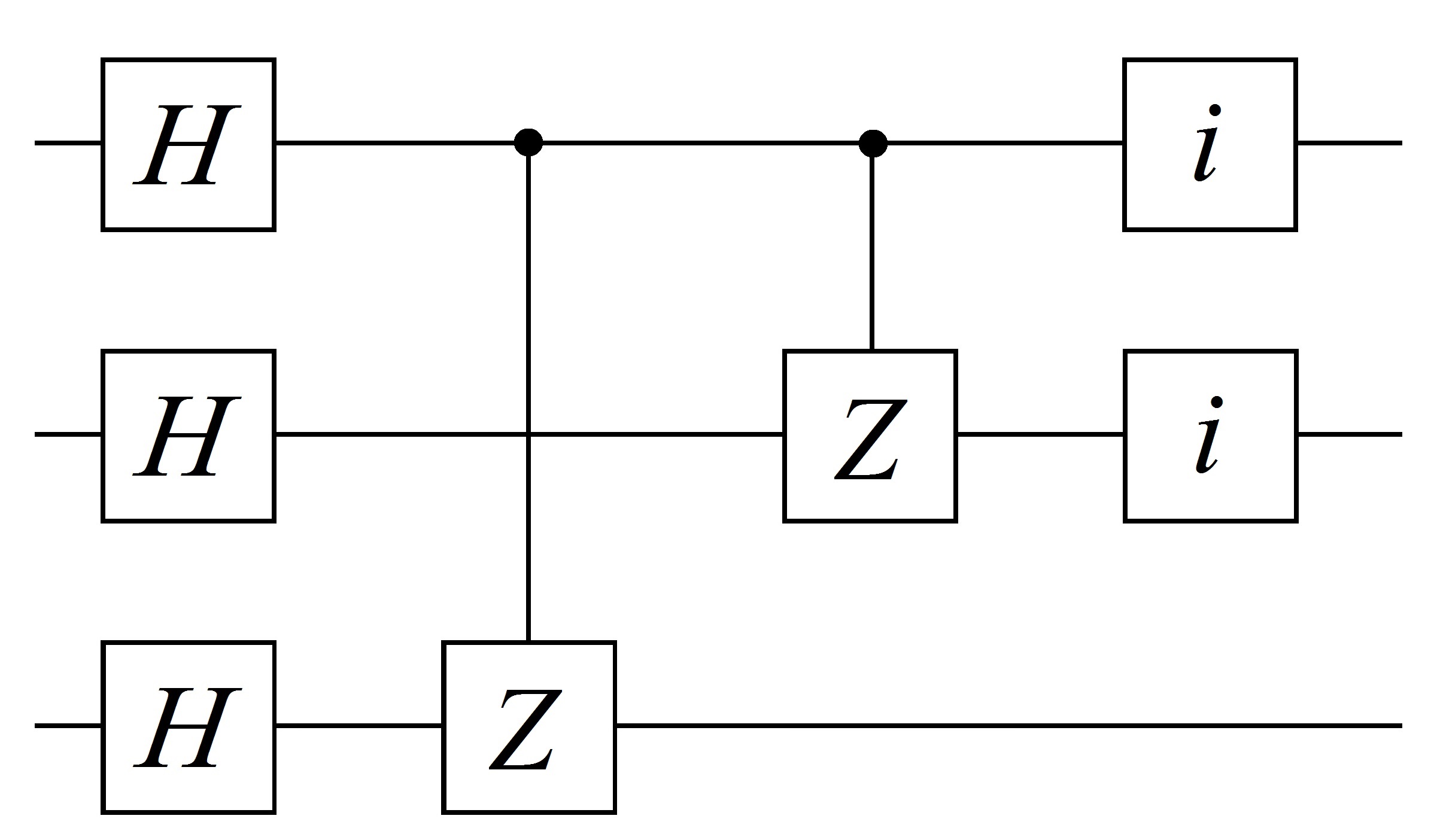}
\caption{Quantum circuit required to construct the fixed set of 9 MUB $\mathcal{S}_4$ \cite{SSL14}.}
\label{CCLuis}
\end{figure}
\begin{center}
QUANTUM CIRCUIT A:\vspace{0.5cm}

\includegraphics[width=6cm]{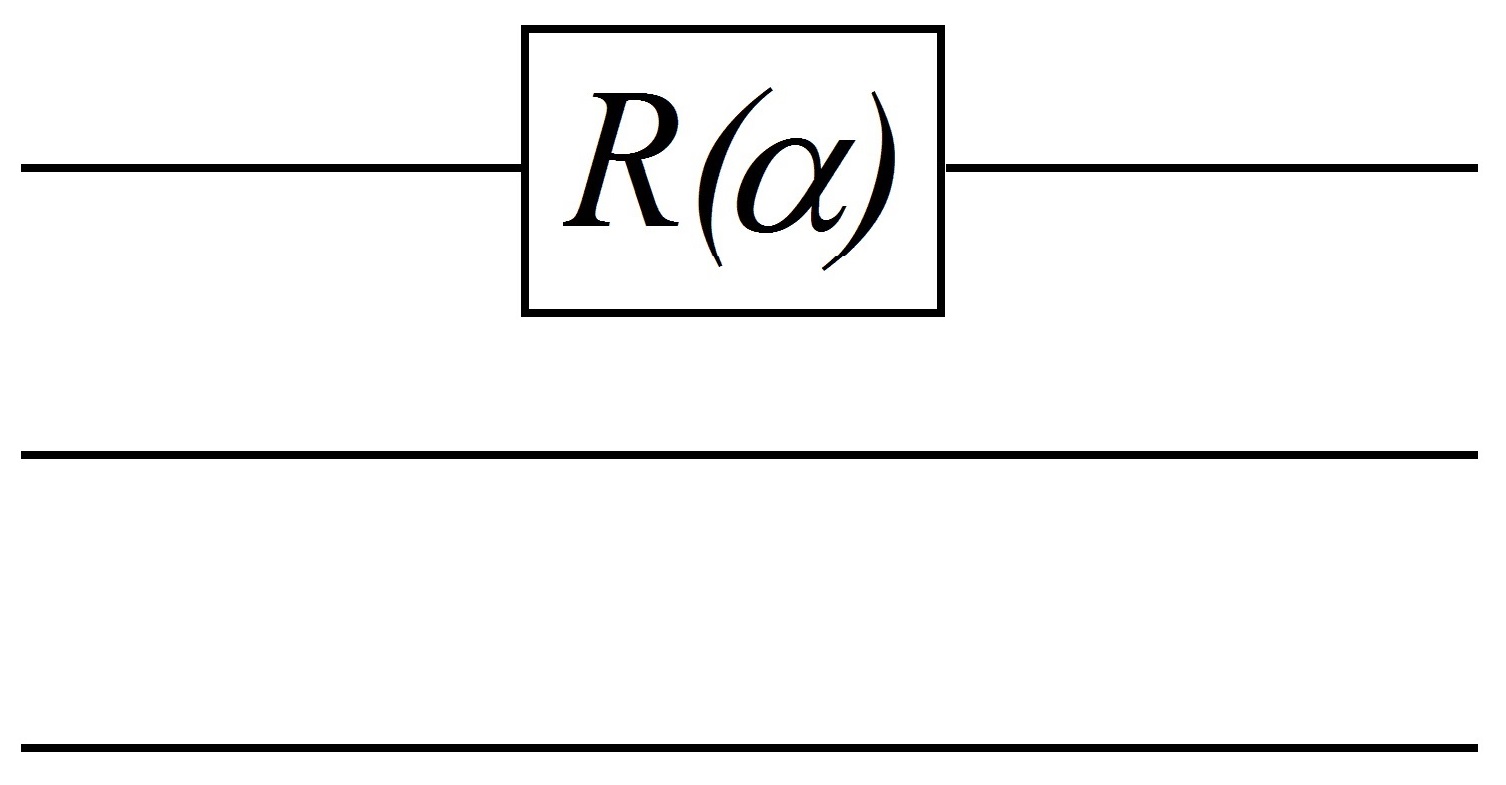}
\end{center}
\begin{center}\vspace{0.5cm}

QUANTUM CIRCUIT B:\vspace{0.5cm}

\includegraphics[width=6cm]{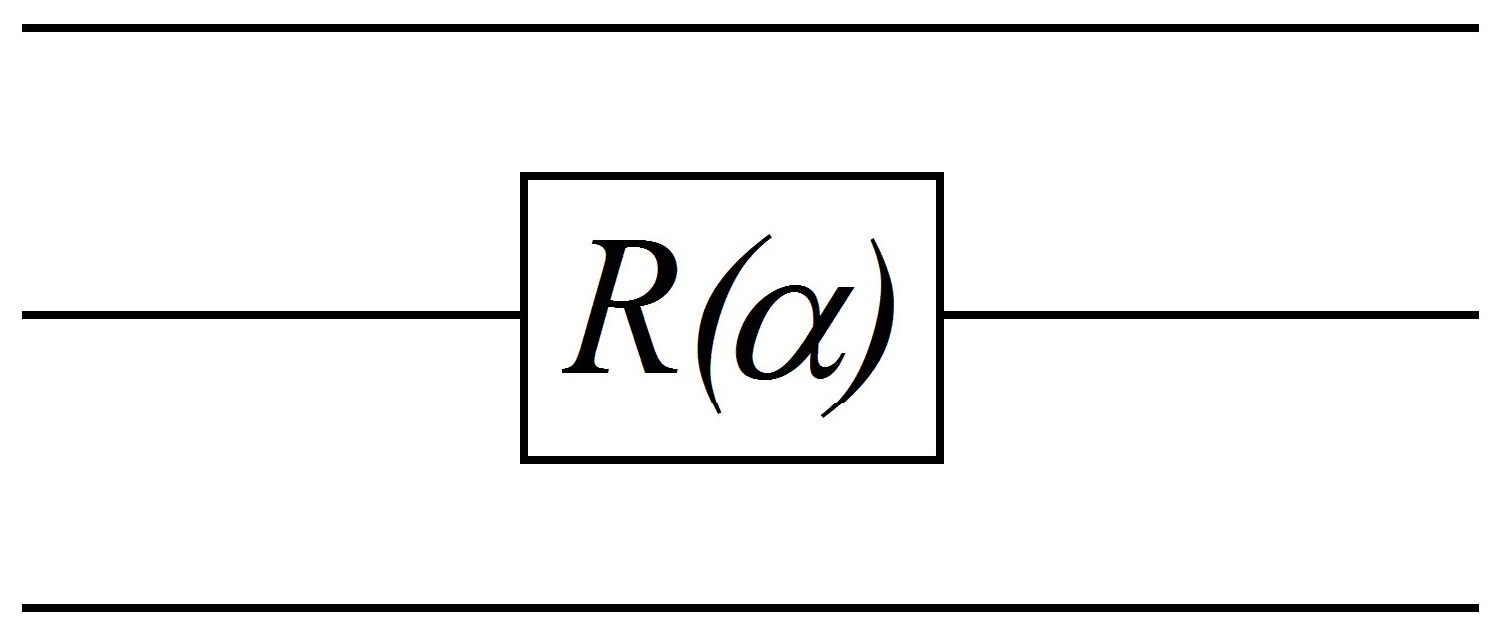}
\end{center}
\begin{center}\vspace{0.5cm}

QUANTUM CIRCUIT C:\vspace{0.5cm}

\includegraphics[width=6cm]{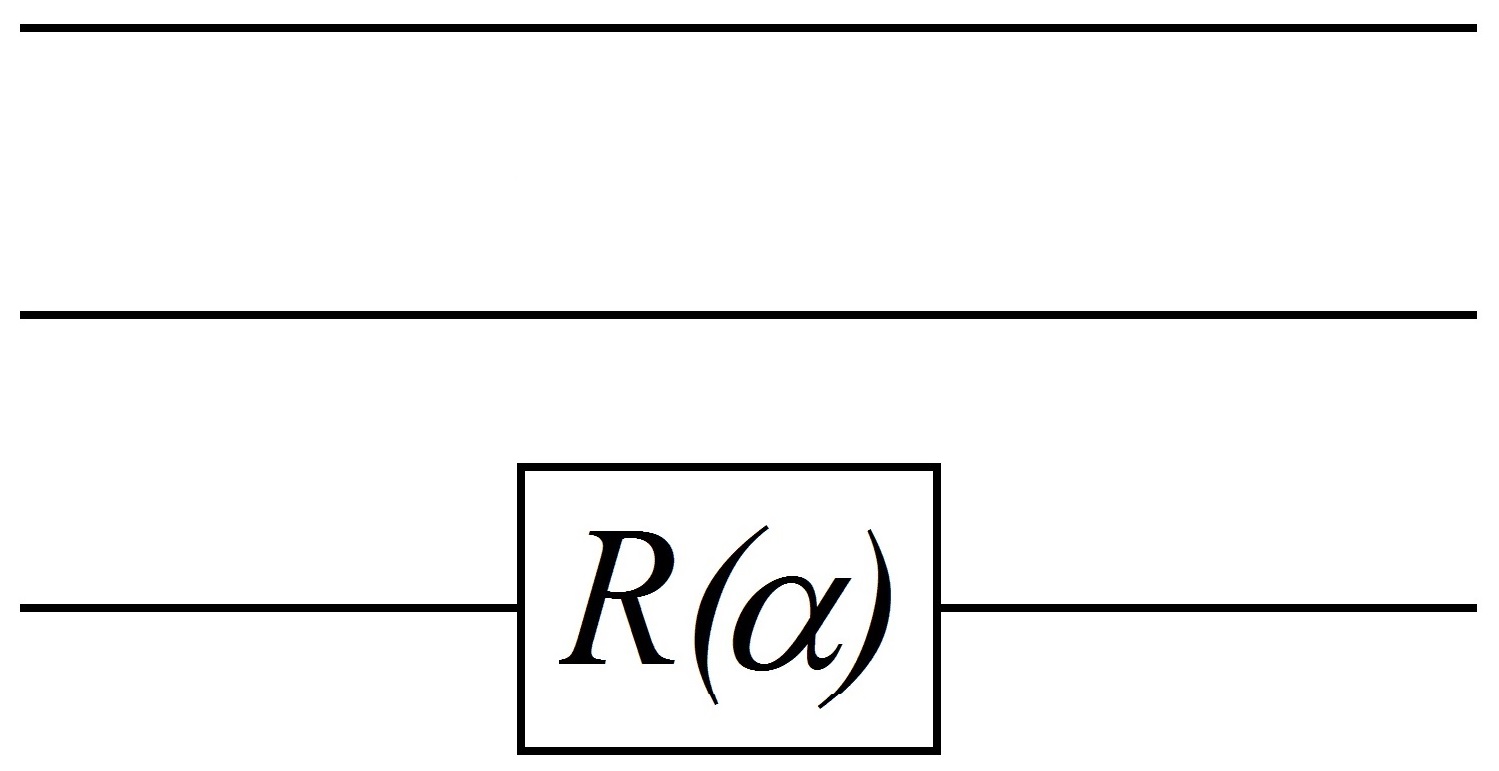}
\end{center}
\begin{center}\vspace{0.5cm}

QUANTUM CIRCUIT D:\vspace{0.5cm}

\includegraphics[width=6cm]{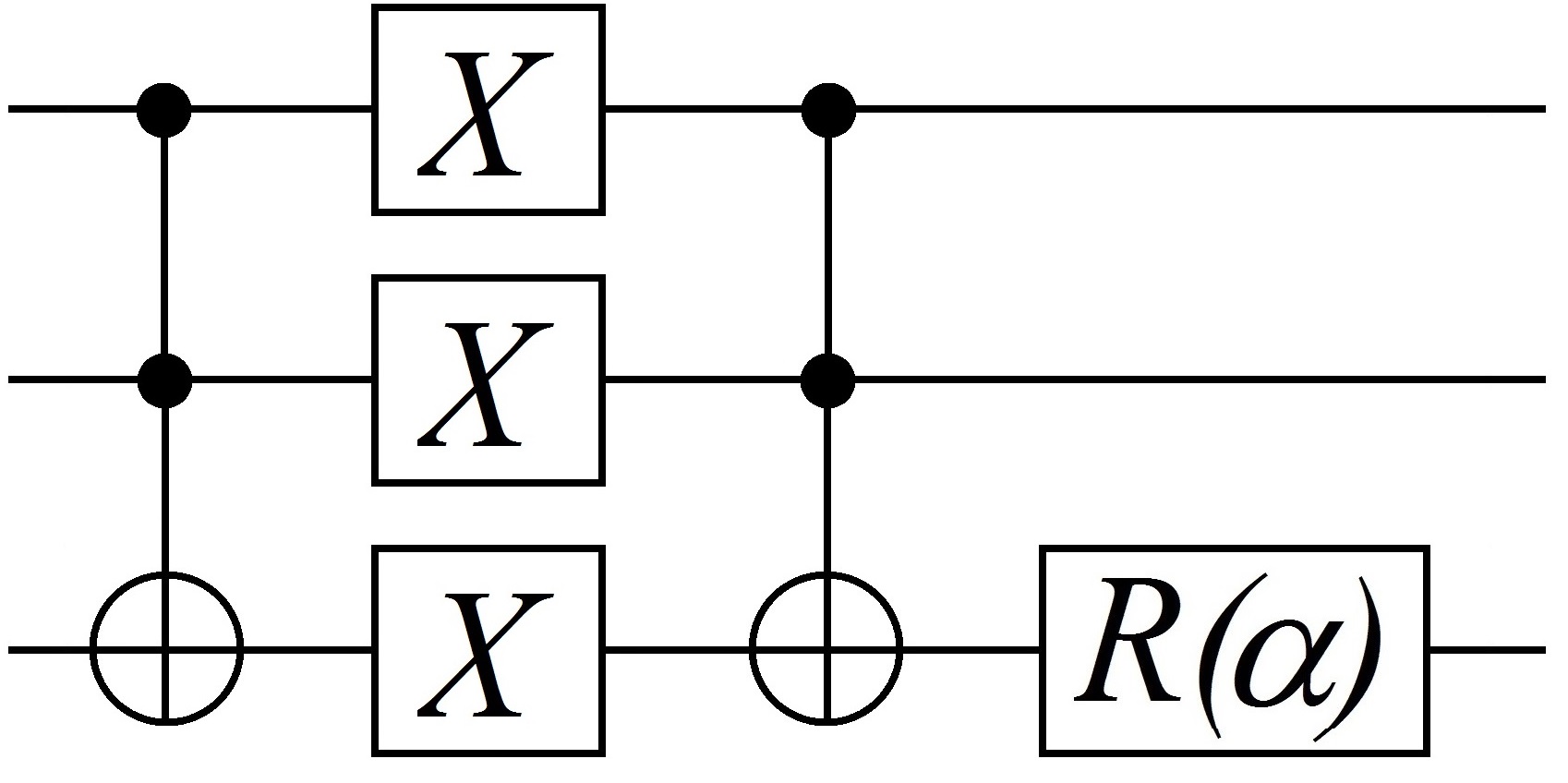}
\end{center}
\begin{center}\vspace{0.5cm}

QUANTUM CIRCUIT E:\vspace{0.5cm}

\includegraphics[width=6cm]{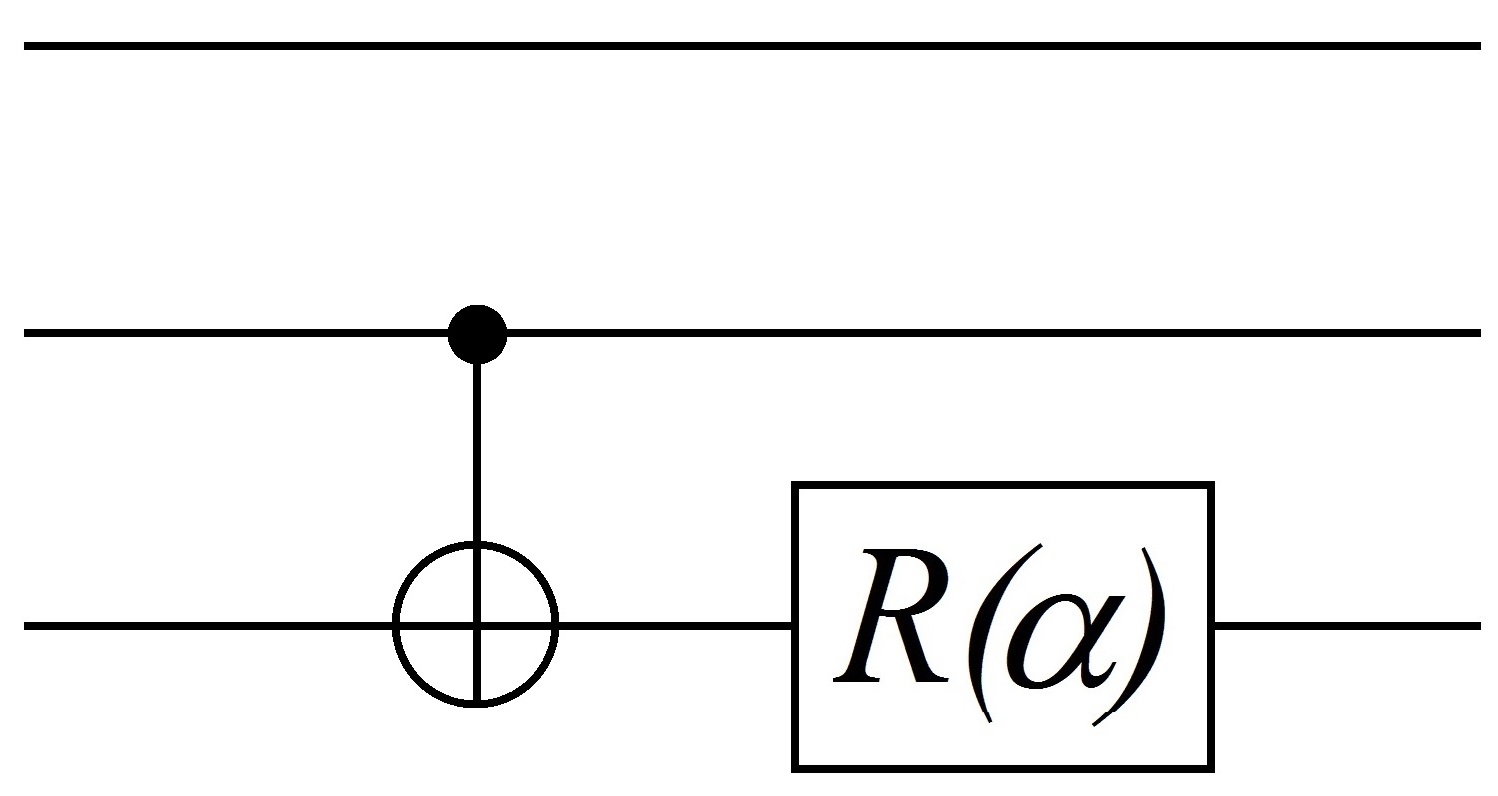}
\end{center}
\begin{center}\vspace{0.5cm}

QUANTUM CIRCUIT F:\vspace{0.5cm}

\includegraphics[width=6cm]{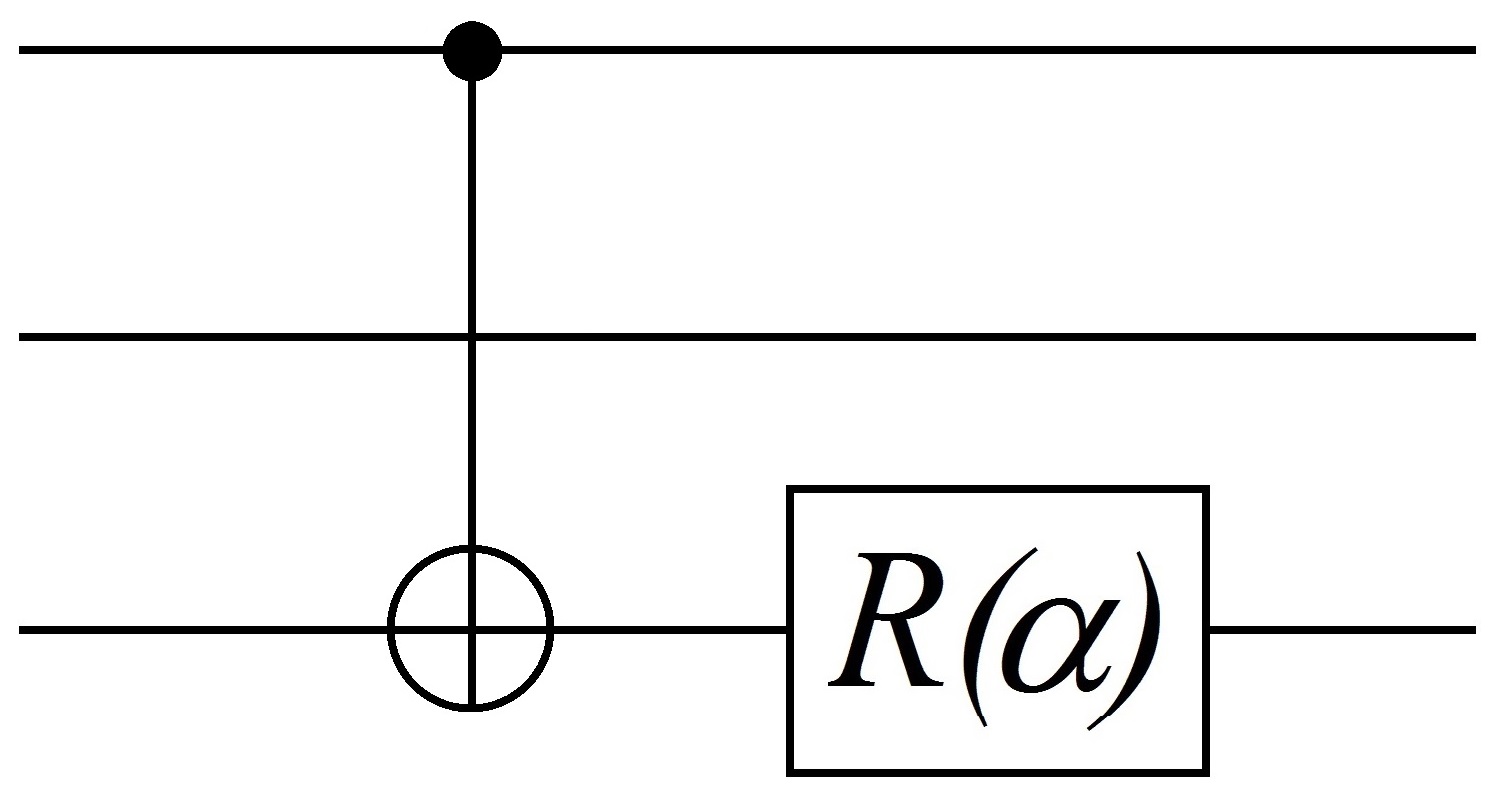}
\end{center}
\begin{center}\vspace{0.5cm}

QUANTUM CIRCUIT G:\vspace{0.5cm}

\includegraphics[width=6cm]{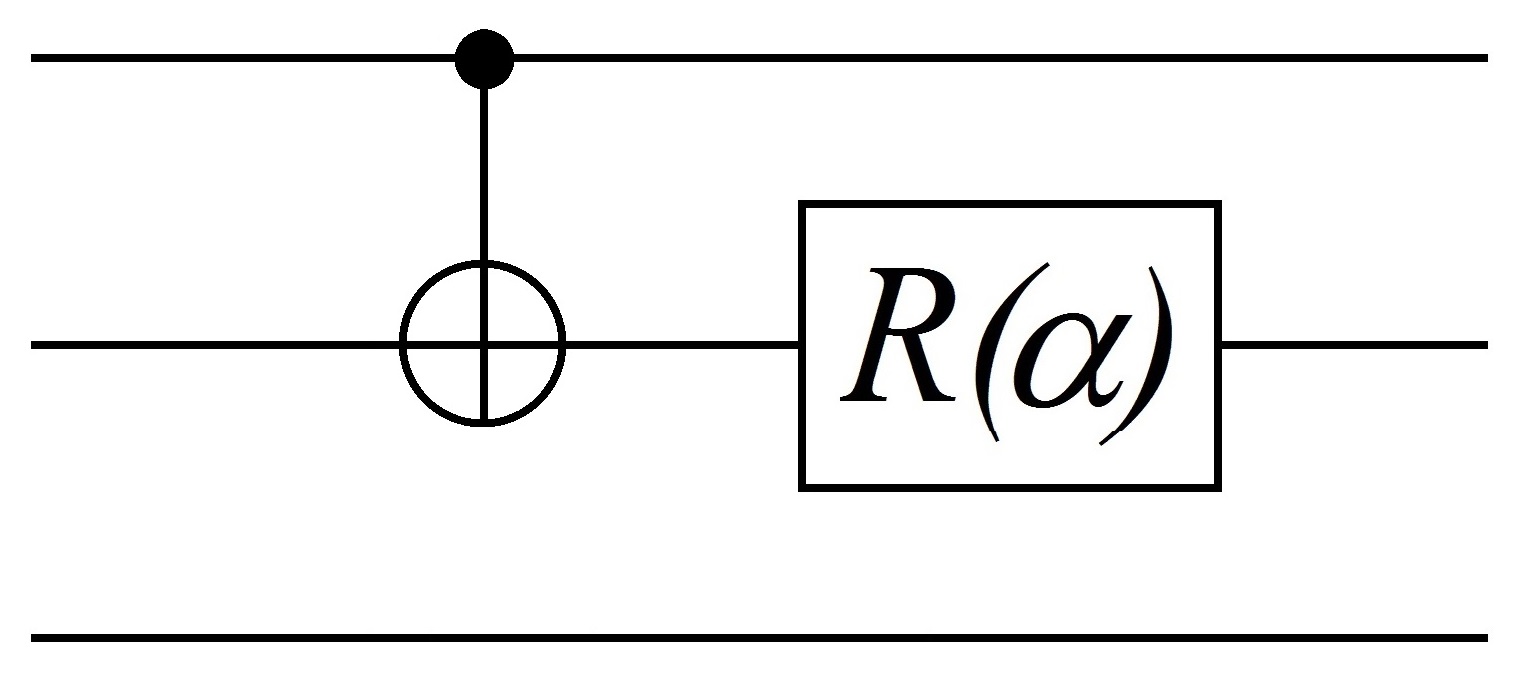}
\end{center}
where
\begin{equation}
R(\alpha)=\left(\begin{array}{cc}
1&0\\
0&e^{i\alpha}
\end{array}
\right),
\end{equation}
and $\alpha\in[0,2\pi)$.
Every quantum circuit (from A to G) introduces aligned free parameters in one of the basis of every quintuplet. The selection of the quadruplet or quintuplet has to be according Table \ref{Tabla2}, as we already explained and exemplified in Appendix \ref{appendix2}. The aligned free parameters appear according the following table:
\begin{equation}\label{circuitrows}
\begin{tabular}{c|ccccccc}
Circuit&A&B&C&D&E&F&G\\ \hline
Rows&1234&1256&1357&1467&2367&2457&3456
 \end{tabular}
\end{equation}
Here, every set of 4 numbers determine the rows where the free parameters are introduced in the bases. For example, the quintuplet $\{\mathbb{I},H_1,H_2,H_3,H_5\}$ provided in Appendix \ref{appendix2} allows the introduction of parameters in $H_1$. The fixed set of 5 bases is generated by the quantum circuit given in Fig.\ref{CCLuis}. According Table \ref{Tabla3} the free parameters can be introduced in $H_1$ by considering the quantum circuit $G$. The aligned parameters in $H_1$ appear in the rows 3,4,5 and 6 (according Table \ref{circuitrows}). We encourage readers to verify these properties from the explicit expressions of the quintuplet $\{\mathbb{I},H_1,H_2,H_3,H_5\}$ provided in Appendix \ref{appendix2}.

\section{Maximal number of MUB: a simple proof for the upper bound}\label{appendix4}
Here, we present an independent proof for the upper bound of the maximal number of MUB in real and complex Hilbert spaces.
\begin{prop}
In dimension $N$ there are $m_R\leq N/2+1$ and $m_C\leq N+1$ MUB for real and complex Hilbert spaces, respectively.
\end{prop}
\begin{proof}
Let $G$ be the Gram matrix of $m$ complex MUB in dimension $N$. Then,
\begin{equation}\label{GG}
G\circ G^{\dag}=\frac{1}{N}\mathbb{J}+\mathbb{I},
\end{equation}
where the $mN\times mN$ matrix $\mathbb{J}$ has $m$ diagonal blocks of size $N$ consisting by the null matrix and the non-diagonal blocks of size $N$ are equal to the unit matrix $\mathbf{1}$ (i.e., every entry of $\mathbf{1}$ is 1). The matrix $\mathbb{I}$ of Eq.(\ref{GG}) is the identity matrix of size $mN$. From matrix theory it is known that given $A,B\geq0$ we have $\mathrm{Rank}(A\circ B)\leq\mathrm{Rank}(A)\mathrm{Rank}(B)$. From considering $A=B^{\dag}=G(n,d)$, the above inequality and Eq.(\ref{GG}) we have $mN-(m-1)\leq N^2$ or, equivalently, $m\leq N+1$. For the real case, we have the tighter inequality $\mathrm{Rank}(A\circ A)\leq\frac{1}{2}\mathrm{Rank}(A)[\mathrm{Rank}(A)+1]$ \footnote{J. Tropp, private communication.}. From combining this inequality with Eq.(\ref{GG}) we have $m_R\leq N/2+1$.
\end{proof}
This proof was inspired in the derivation of the upper bound of the maximal number of vectors in Equiangular Tight Frames \cite{STDH07}.

\end{document}